\documentclass[aps,pra,showpacs,twoside,twocolumn,10pt]{revtex4-1}
\usepackage[colorlinks=true, citecolor=cyan, urlcolor=teal ]{hyperref}
\usepackage{times,epsfig,amssymb,amsfonts,amsmath,bm,subfigure,mathtools,amsthm,braket,soul,enumitem,xcolor,physics,graphics,graphicx}
\usepackage[normalem]{ulem}

\newtheorem{definition}{Definition}
\newtheorem{proposition}{Proposition}

\newtheorem{observation}{Observation}

\begin{document}
\title{Operational Ergotropy: suboptimality of the geodesic drive
}
\author{Pritam Halder$^{1}$, Srijon Ghosh$^{1}$, Saptarshi Roy$^{2}$, Tamal Guha$^{2}$}
\affiliation{$^1$ Harish-Chandra Research Institute, A CI of HBNI, Chhatnag Road, Jhunsi, Prayagraj - 211019, India \\ 
$^2$ QICI Quantum Information and Computation Initiative, Department of Computer Science,
The University of Hong Kong, Pokfulam Road, Hong Kong
}

\begin{abstract}

We put forth a notion of optimality for extracting ergotropic work, derived from an energy constraint governing the necessary dynamics for work extraction in a quantum system. Within the traditional ergotropy framework, which predicts an infinite set of equivalent pacifying unitaries, we demonstrate that the optimal choice lies in driving along the geodesic connecting a given state to its corresponding passive state. Moreover, in a practical scenario where unitaries are inevitably affected by environmental factors, we refine the existing definition of ergotropy and introduce the notion of operational ergotropy. It enables the characterization of work extraction in noisy scenarios. We find that for certain typical noise models, the optimal choice  which governs the Schrödinger part of the dynamics, aligns with the optimal drive in the unperturbed scenario. However, we demonstrate that such optimality is not universal by presenting an explicit counterexample.
Additionally, within this generalized framework, we discuss the potential for faster work extraction from quantum systems in the presence of noise.

 \end{abstract}

\maketitle

\section{Introduction}



In the realm of quantum thermodynamics, a fundamental pursuit involves the extraction of energy from a given quantum state, which serves as a temporary reservoir of energy, facilitating its subsequent transfer at a later time—a cornerstone objective for thermal machinery. Work extraction in quantum thermodynamics can be broadly classified into two categories. First with respect to a heat bath \cite{scully2003,michal2013,sandu2014} and the other one, which would be the quantity of interest in this letter, is ergotropy \cite{Allahverdyan_2004}. For a reference Hamiltonian, the maximum reversible work achievable via unitary evolution \cite{Lenard1978,Pusz1978} is termed as ergotropy \cite{Allahverdyan_2004}. Ergotropy is considered to be a pivotal concept in the development of efficient quantum batteries \cite{alicki_2013, gian2018, gian2019, donato2019, ghosh2021}. Along with protocol for work extraction from unknown state \cite{rosa2023_unknown}, extensive investigation has been undertaken to amplify ergotropy and optimize energy transfer by leveraging various quantum attributes including entangling operations \cite{alicki_2013,Hovhannisyan2013,Giorgi2015,Perarnau-Llobet2015}, quantum coherence \cite{francica2020,cakmak2020}, quantum correlations \cite{Perarnau-Llobet2015_correl,Binder2015,salvia2022,Touil2022,francica2022}, and measurement-feedback control \cite{Francica2017,koshihara2023}. Moreover, from the operational perspective, ergotropic gap \cite{marti2015} has been shown to be a witness \cite{mir2019} of quantum entanglement \cite{HHH2009}. Recently, by following the avenue of of local work extraction protocols, the notion of extended local ergotropy \cite{castellano2024} and work extraction via non-completely positive trace-preserving \cite{bhattacharyya2024} has also been studied. Despite all these 
research, a fundamental conceptual void regarding the amount of time required for work extraction remained unaddressed which we discuss in this letter.

The standard definition of ergotropy primarily focuses on optimizing work extraction, neglecting the requisite potential energy-time relation to operationalize the pacifying unitary in an experimental scenario. In particular, there should be a bound imposed on the strength of the potential that drives the system towards work extraction. The importance of an energy constraint comes from the fact that realistically there is only a finite amount of energy that is available in an experimental setting. The energy-constrained potential must guide the system until a specific temporal threshold, a consideration that has consistently remained overlooked within the established framework of ergotropy. With explicit time-dependence of work extraction, it can further provide additional insights into the effects of the environment during the process of driving the non-passive initial state which holds significance in fortifying the robustness of ergotropy. 

In this work, we incorporate the time explicitly during the cyclic work extraction process. This further shows that the pacifying unitary evolution is non-unique, i.e., there exist infinite possible planes in which the driving can be performed. However, the optimal time required to attain maximum work extraction, i.e., ergotropy, corresponds solely to the geodesic drive. Furthermore, our investigation delves into the impact of decoherence on ergotropy. The motivation for the same stems from the practical observation that the process of extracting work may encounter environmental perturbations in real-world scenarios instead of a closed unitary process. To address this we examine three prototypical noise models 
and show that although the value of ergotropy decreases in noisy time evolution there exists a specific type of noise that can facilitate reduction in optimal time required and also increase output power in some scenarios. 

The pinnacle of our study pertains to the observation that while the geodesic drive stands as an optimal strategy for extracting work within the context of above mentioned noise models, it may not invariably represent the most efficient approach for minimizing time in maximum work extraction scenarios. Our work offers a proof of principle approach in supporting this notion by demonstrating the presence of a dynamical process governed by a completely positive and trace preserving (CPTP) map, wherein the work extraction along geodesic path falls short of achieving optimality.

The paper is structured as follows. In Sec.~\ref{sec:2}, we address the inherent limitations in the conventional concept of work-extraction, which we modify by invoking explicit time dependence in the definition of ergotropy. This modification is utilized to confirm the non-equivalence of pacifying unitaries, thereby facilitating the derivation of optimal pacifying unitaries. In Sec.~\ref{sec:3}, we introduce the notion of \textit{operational ergotropy}, motivated from the experimental perspective where the dynamics governed by the pacifying unitary is hindered by the inevitable presence of environment during practical implementation. Moreover, in Sec.~\ref{markovian}, we demonstrate that akin to the noiseless scenario, employing geodesic drive proves optimal in the context of the known noise models, offering an additional advantage in accelerating the work extraction process. Conversely, in Sec.~\ref{sec:geodesic}, we prove that the optimality of geodesic drive is not universal. Finally, we conclude in Sec.~\ref{conclusion}.

\section{Introducing optimality in the standard notion of ergotropy}
\label{sec:2}
Before laying out the limitations of ergotropic work extraction and arguing for a generalization, let us first analyze the canonical formulation of ergotropic work extraction. 
In the literature, the standard notion of ergotropy is defined as the maximum amount of work that can be extracted from a quantum system \cite{Allahverdyan_2004}. Typically, the average energy of a $d$-dimensional quantum system $\rho$ is defined for a reference Hamiltonian $H$ as $\langle E \rangle = \Tr(\rho H)$. On the other hand, the ergotropy of the system $\rho$
is defined as 
\begin{eqnarray}
    \mathcal{W} = \Tr(\rho H) - \min_{U}\Tr(U \rho U^{\dagger} H),
    \label{eq:ergotropydef1}
\end{eqnarray}
where minimization has to be done over all possible unitary operations. 
\begin{figure*}
\includegraphics[width=\textwidth]{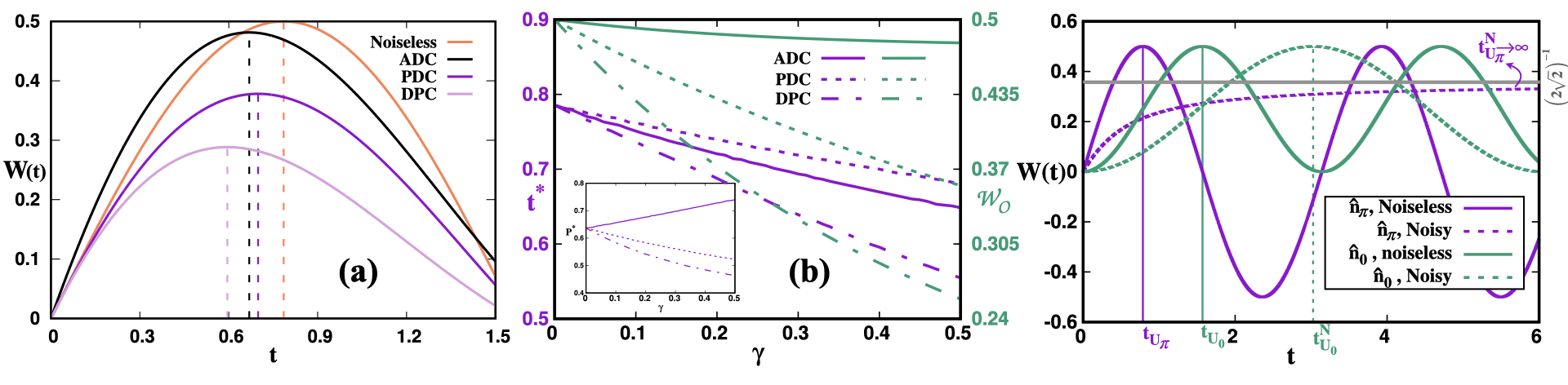}
\caption{(Color online.) The initial state is taken to be $\ket{+}$ with $\omega_{\max}=2$. In (a) we plot the extracted work, $W(t)$ (vertical axis) as a function of time, $t$ (horizontal axis) for the noiseless and some prototypical   Markovian noise models, namely amplitude damping channel (ADC), phase damping channel (PDC), and depolarizing channel (DPC). The decay rate $\gamma=0.4$ for all the noise models. In (b) we plot the minimum time, $t^*$ (left ordinate) for maximum work $\mathcal{W_O}$ (right ordinate) extraction with the decay rates $\gamma$ (abscissa). Both $t^*$ and $\mathcal{W_O}$ decreases with $\gamma$. Interestingly, we show in the inset of (b), the output power, $P^*$ (vertical axis) for the ADC increases with increase in the value of $\gamma$ (horizontal axis) unlike the case for PDC and DPC where the power decreases monotonically. In (c), the extracted work, $W(t)$ (ordinate) is plotted as a function of time, $t$ (abscissa) for two different pacifying unitaries in the presence and absence of noise described in Sec.~\ref{sec:geodesic}. The solid curves correspond to the noiseless case. The solid violet curve corresponds to the geodesic drive by rotation around $\bm{\hat{n}_{\pi}}=(0,-1,0)$, i.e., negative $y$ axis, while the green solid curve corresponds to pacification via rotation along the axis $\bm{\hat{n}_0} = \frac{1}{\sqrt{2}}(1,0,1)$. In the presence of the noise with $\zeta = 5$, the green dashed curve denotes the case where the rotation along the axis $\bm{\hat{n}_0}$ is employed while the violet dashed one corresponds to the geodesic drive. Note that, in the noiseless scenario, $t_{U_{\pi}}<t_{U_0}$ whereas in the noisy evolution the situation is reversed, i.e., $t_{U_{\pi}}^N>t_{U_0}^N$. All the axes are dimensionless.}
\label{fig:alpha_space}
\end{figure*}
The minimization requires the unitary to convert $\rho$ to its corresponding passive state $\rho_{p} = U \rho U^\dagger$ \cite{Lenard1978,Pusz1978}, from which no work can be extracted via any unitary cycle, hence we call such a unitary to be a \textit{pacifying unitary}. 
The pacifying unitaries are generated by turning on an external potential $V(t)$ for a time $\tau$ such that
\begin{eqnarray}
    U(\tau) = e^{-i\int_0^\tau H(t) dt},
    \label{eq:pu}
\end{eqnarray}
where $H(t) = H+V(t)$.
Without loss of generality, one can take the external potential to be constant during the duration of the work extraction process \cite{Binder2015, alicki_2013, Tamal2020}, i.e., $V(t) = V$ for $0 \leq t \leq \tau$ and $0$ elsewhere. 

To find out the form of the pacifying unitary  and the passive state $\rho_{p}$, 
one needs to decompose $\rho$ and $H$ in their spectral form as $\rho = \sum _{k = 1}^{d} r_{k} |r_{k}\rangle \langle r_{k}|$ and $H = \sum _{k = 1}^{d} \epsilon_{k} |\epsilon_{k}\rangle \langle \epsilon_{k}|$ respectively. Here $\{|r_{k}\rangle\}_{k}$ and  $\{|\epsilon_{k}\rangle\}_{k}$ are the eigenvectors and $\{r_{k}\}_{k}$ and  $\{\epsilon_{k}\}_{k}$ are the properly ordered eigenvalues of the $\rho$ and $H$ respectively, such that, $r_{k} \geq r_{k+1}$ and $\epsilon_{k} \leq \epsilon_{k+1}$ . Since the passive state can be mathematically written as $\rho_{p} = \sum_{k = 1}^{d} r_{k} |\epsilon_{k}\rangle \langle \epsilon_{k}|$ \cite{Allahverdyan_2004, farina2019}, by construction passive state energy becomes $\sum_{k} r_{k} \epsilon_{k}$. 
Interestingly the choice of the pacifying unitary is not unique. The set of pacifying unitaries is denoted by $\{ U_{\{\mu_k\}}\}$ where $k$ runs up to the system dimension, and each $\mu_k$ is a phase: 
\begin{eqnarray}
U_{\{ \mu_k \}} = \sum_k e^{i \mu_k} \ketbra{\epsilon_k}{r_k},
\label{eq:pugeneral}
\end{eqnarray}
with $0 \leq \mu_k < 2\pi$ $\forall k$, and without loss of any generality we can take $\mu_0 = 0$. Therefore we have a collection of unitaries that results in pacification. Now, the existing framework of ergotropy indicates that all these unitaries are equivalent as they extract the same amount of work. Till date, the role of these phases was mostly assumed to be irrelevant. However, in a realistic energy-constrained scenario, where the maximal strength of the engineered Hamiltonians is limited, it can be shown that all these unitaries take different \textit{times} for pacification, which we denote by $t_{U_{\{ \mu_k \}}}$. Consequently, this variation in temporal requirement lifts the degeneracy in their status. The energy constraint is imposed by the following condition on the Hamiltonian generating the pacifying unitary in Eq. \eqref{eq:pu} as
\begin{eqnarray}
    ||H(t)||_1 \leq \omega_{\max},
    \label{eq:energybound1}
\end{eqnarray}
where $||A||_1 = \text{Tr} ~\sqrt{AA^\dagger}$ denotes the trace norm of $A$. Physically, this limits the maximal amount of energy that can be pumped into the system. Intuitively, we expect this bound to translate to 
\begin{eqnarray}
    t^* \sim \frac{1}{||H(t)||_1} \sim \frac{1}{\omega_{\max}},
\end{eqnarray}
where $t^*$ is the minimum pacification time consistent with the energy bound in Eq. \eqref{eq:energybound1}, such that $t^* = \min_{\{\mu_k\}} t_{U_{\{ \mu_k \}}}$. Later we would demonstrate the relationship between $t^*$ and $\omega_{\max}$ explicitly. The requirement for the minimal pacification time leads to an emergent notion of optimality.

\begin{definition}
    The optimal pacifying unitary $U_{\{\mu_k^*\}} \in \{ U_{\{\mu_k\}}\}$ is the one that extracts work equal to the ergotropy in the minimum amount of time.
\end{definition}

Now the central question is, given an energy bound as in Eq. \eqref{eq:energybound1}, what is the optimal pacifying unitary? We address this in the following Proposition.
\begin{proposition}
    The optimal pacifying drive is along the geodesic connecting the states  $\rho$ to its corresponding passive state $\rho_p$. 
\label{lemma:1}
\end{proposition}
\begin{proof}
As mentioned before, during the work extraction period which lasts for a time $\tau$, the driving Hamiltonian $H + V$ generates unitaries that drive $\rho$ to $\rho_p$ via the various trajectories with the same constant speed respecting the energy constraint. Therefore, the evolution that pacifies $\rho$ in the minimal time corresponds to the shortest path connecting $\rho$ and $\rho_p$ on the state space, referred to as the \textit{geodesic}.
\end{proof}


Therefore, in the noiseless case, 
where Proposition.~\ref{lemma:1} uniquely restricts $U_{\{\mu_k^*\}}$ to be the geodesic drive, the complete physics of work extraction for a given energy bound is characterized by the triad: $(\mathcal{W}, t^*, U_{\{\mu_k^*\}})$. Here, once $U_{\{\mu_k^*\}}$ is specified, $t^*$ gets fixed automatically. However, in a more generalized setting involving noise, Proposition. \ref{lemma:1} does not hold and hence the role of the triad to specify work extraction becomes more involved. The importance of such a characterization will be reemphasized subsequently.

Let us now demonstrate how the optimal pacifying unitary can be evaluated for two-dimensional quantum systems where the
geometry of the state space is much cleaner.
Consider an arbitrary qubit state $\rho$, specified by spherical polar coordinates $(r,\theta_0,\phi_0)$ of the Bloch sphere. The spectral form of the state is given by
\begin{eqnarray}
\rho = \lambda |\psi\rangle \langle \psi| + (1-\lambda)|\psi^\perp \rangle \langle \psi^\perp |,
\label{eq:rho}
\end{eqnarray}
where 
without loss of any generality, $\lambda = \frac{1+r}{2} \geq \frac12$, and $\ket{\psi} = \cos \frac{\theta_0}{2} \ket{0} + e^{i \phi_0} \sin \frac{\theta_0}{2} \ket{1}$.
The Hamiltonian with respect to which energies are measured is taken to be  $H = \ketbra{1}{1} =  \frac12 (\mathbb{I} - \sigma_3)$.
Here, $\sigma_0= \mathbb{I}$ and $\sigma_1 = \sigma_x, \sigma_2 = \sigma_y, \sigma_3 = \sigma_z$ are the Pauli matrices. The passive state of $\rho$ takes the form $\rho_p = \lambda |0 \rangle \langle 0| + (1-\lambda)|1 \rangle \langle 1 |$, which can be specified by $(r,0,0)$ in spherical polar coordinates.  Now, the ergotropy from Eq. \eqref{eq:ergotropydef1} is found to be $  \mathcal{W} = r \sin^2 \frac{\theta_0}{2}$.
The corresponding pacifying unitaries, following Eq. \eqref{eq:pugeneral}, can be defined by a one-parameter family
\begin{eqnarray}
    U_\mu = |0\rangle \langle \psi | + e^{i \mu}|1\rangle \langle \psi^\perp|, 
    \label{eq:genpaciuni}
\end{eqnarray}
where $0 \leq \mu < 2\pi$, and $\ket{\psi}, \ket{\psi^\perp}$ are the eigenvectors of $\rho$ (see Eq. \eqref{eq:rho}).

From Proposition. \ref{lemma:1} we know the optimal pacifying unitary evolves an arbitrary state along the geodesic connecting itself to its passive state. To evaluate it, we first provide a geometric explanation for the infinite cardinality of $\{U_\mu \}$ which contains pacifying unitaries corresponding to each allowed value of $\mu$. Recall that a unitary evolution between two points on the Bloch sphere of identical purity follows a circular path. This trajectory is created by the intersection of a plane connecting $\rho$ and $\rho_p$ with the Bloch sphere. Since there are infinitely many planes with this property we have the set $\{U_\mu \}$. Each value of $\mu$ corresponds to a different pacification trajectory and thereby a different plane.
 Therefore, the optimal pacification corresponds to a driving Hamiltonian whose direction  $\bf{\hat{n}_{\mu^*}}$ is normal to the plane of the geodesic, i.e.,
 \begin{eqnarray}
 {\bf{\hat{n}_{\mu^*}}} = \{ \sin \phi_0, -\cos \phi_0,0\},
 \label{eq:unitnormal}
\end{eqnarray}
(see Appendix \ref{app:B} for details). This in turn produces the optimal pacifying unitary $U_{\mu^*}\in \{U_\mu\}$, where $\mu^* = \phi_0 + \pi\mod2\pi$.
Note that $\bf{\hat{n}_{\mu^*}}$ resides entirely in the $x-y$ plane. Consequently, the target pacifying unitary is a rotation of angle $\theta_0$ about the axis $\bf{\hat{n}_{\mu^*}}$, i.e., $U_p \equiv U_{\mu^*} = \exp \big(-i (\theta_0/2) {\bf{\hat{n}_{\mu^*}.}} \bm{\sigma}\big)$. 
On the other hand, the driving Hamiltonian can be written as  $H(t) = H + V_{\mu^*} = K {\bf{\hat{n}_{\mu^*}.}} \bm{\sigma}$ $(K>0)$, where $\bm{\sigma}$ is the vector of the Pauli matrices $\{ \sigma_1, \sigma_2, \sigma_3 \}$. Note that in principle $K$ can be arbitrarily large. 
However, following the energy bound in Eq. \eqref{eq:energybound1} we get $2K \leq \omega_{\max}$.
Therefore, the fastest driving compatible with the energy bound is when $K = \omega_{\max}/2$ and the corresponding driving unitary is $ U(t) = \exp \big(-i(\omega_{\max}/2)t ~{\bf{\hat{n}_{\mu^*}.}} \bm{\sigma}\big)$. 
Finally, by comparing, i.e., $U(t^*) = U_p$, we get $t^* = \frac{\theta_0}{\omega_{\max}}$ which is the minimal time required to extract the ergotropy consistent with the energy bound given in Eq. \eqref{eq:energybound1}. In general, the time $\tau$ for drawing out the ergotropy from the system respects the following speed limit 
\begin{equation}
    \tau \geq t^*= \frac{\theta_0}{\omega_{\max}}.
    \label{eq:speedlim1}
\end{equation}
To conclude, we show that the complete characterization of the optimal work extraction from the single qubit systems is specified by triad  $\left(r\sin^2\frac{\theta_0}{2},\frac{\theta_0}{\omega_{\max}}, U_{\phi_0+\pi\mod2\pi}\right)$.



Let us now illustrate the choice of the optimal unitary with a simple example. Suppose the initial state of the system is $\rho = |+\rangle\langle+|$ and for which the passive state is $\rho_{p} =|0\rangle\langle0|$. The set of pacifying unitaries from Eq. \eqref{eq:genpaciuni} is given by $U_{\mu} = \ketbra{0}{+} + e^{i \mu} \ketbra{1}{-}$. Now consider two pacifying unitaries among the set $\{U_{\mu}\}$ indexed by $\mu = 0$ and $\pi$ which can be written as $U_{0} = (\sigma_{1} + \sigma_{3})/\sqrt{2}$ and $U_{\pi} = (\mathbb{I} + i \sigma_{2})/\sqrt{2}$. Here, $U_{0}$ is a rotation about $\bm{\hat{n}}_0 = \frac{1}{\sqrt{2}}(1,0,1)$ upto an overall phase, whereas $U_{\pi}$ indicates the dynamics along the geodesic connecting $\rho$ and $\rho_{p}$ on the Bloch sphere. Without loss of generality,    by choosing $\omega_{\max} = 2$, the pacification times corresponding to the above mentioned unitaries reads as $t_{U_{0}} = \frac{\pi}{2\sqrt{2}}$, and $t_{U_{\pi}} = \frac{\pi}{4}=t^*$.
It is straightforward to show that $U_\pi$ offers the minimal time of work extraction saturating Eq. \eqref{eq:speedlim1}. Therefore, among the set $\{ U_\mu \}$, $U_\pi$ turns out to be the optimal pacifying unitary.



\section{Operational Ergotropy}
\label{sec:3}
We are now equipped with the optimal pacifying unitary for extracting the maximum amount of work from a quantum system in the minimum amount of time respecting a given energy constraint. Nevertheless, from a realistic standpoint, achieving precise unitary implementations will always be hindered due to the inevitability of noise. In this section, we introduce the concept of a generalized notion of ergotropy, namely \textit{operational ergotropy}, where the work extraction is governed by \sout{some} noisy unitary dynamics. 
\begin{definition}
\label{def:2}
 The operational ergotropy $\mathcal{W}_O$ is defined by
\begin{equation}
  \mathcal{W}_O := \max_{t,\{\mu_k\}} W_{\{\mu_k\}}(t),
  \label{eq:dynamicalergotropy}
\end{equation}
with $W_{\{\mu_k\}}(t) = \text{Tr} (H \rho) - \text{Tr} (H \Lambda_t^{\{\mu_k\}}(\rho))$. Here $\Lambda_t^{\{\mu_k\}}$ denotes the dynamical map that takes $\rho \to \rho_t$.  
\end{definition}
 As before, $\{\mu_k\}$ represents the pacifying unitary that controls the Schrödinger part of the dynamics, see Eq. \eqref{eq:pugeneral}. Naturally, the first course of action is to investigate how noise affects the total maximal work that can be extracted compared to the noiseless case and how the time in which this work can be extracted gets altered.
The other essential point of interest is whether the optimal choice of $\{\mu_k^*\}$ coincides with the geodesic drive, optimal in the noiseless case. Physically, it means whether it is sufficient to use the pacifying unitary that was optimal in the noiseless situation, or whether it is preferable to use a different pacifying unitary (see Eq. \eqref{eq:genpaciuni}) that produces better work extraction features in the presence of noise. 
These questions can be addressed by analyzing the generalized triad $(\mathcal{W}_O,t^*, U_{\{\mu_k^*\}})$ where the extracted work is calibrated using the operational ergotropy. Interestingly, due to the presence of perturbations, $U_{\{\mu_k^*\}}$ alone cannot determine the optimal time in this scenario unlike in the noiseless case.
It is worth noting that when discussing the implementation of a unitary in the presence of noise, we specifically refer to its Hamiltonian generator, responsible for governing the Schrödinger part of the dynamics.

\subsection{Work extraction in the presence of paradigmatic noises}
\label{markovian}
Let us first consider the noise to be a memoryless channel, i.e., Markovian noise. The action of the channel on the initial state $\rho$ can be represented as $ \Lambda_t^{\{\mu_k\}}(\rho) = e^{\mathcal{L}_{\{\mu_k\}} t} \rho,$
where the Lindbladian $\mathcal{L}_{\{\mu_k\}}$ satisfies $\dot{\rho} = \mathcal{L}_{\{\mu_k\}} \rho  = -i [H + V_{\{\mu_k\}},\rho] + \mathcal{D}(\rho)$. Here, $H + V_{\{\mu_k\}}$ is the generator of $U_{\{\mu_k\}}$ in Eq. \eqref{eq:pugeneral}. $\mathcal{D}$ is the dissipator for the corresponding Markovian noise. For the sake of simplicity, we will again demonstrate our findings with two-level quantum systems (See Appendix \ref{app:C} for details). Specifically, we again consider $\rho = \ketbra{+}{+}$ to be the system's initial state, and $\omega_{\max} = 2$. We analyze the effect of amplitude damping channel (ADC), phase damping channel (PDC), and depolarizing channel (DPC) (see Appendix \ref{master_eqn} for the form of their dissipators) in the process of work extraction and argue whether the optimality of the geodesic drive still holds. 

\subsubsection{Noise speeds up work extraction}

In the presence of noise, the operational ergotropy $\mathcal{W}_O$ decreases with an increase in the noise strength, $\gamma$, for all the noise models under consideration. Moreover, the extent of reduction is inherently dependent upon the specific nature of the noise (Fig.~\ref{fig:alpha_space}(a) and (b)). The next quantity of interest from the triad $(\mathcal{W}_O,t^*, U_{\{\mu_k^*\}})$ is the optimal time, $t^{*}$. It is the earliest instant at which the extracted work attains its maximum value. Mathematically, $t^*$ is the temporal maximizer in the definition of operational ergotropy in Eq.~\eqref{eq:dynamicalergotropy}. 
Alongwith $\mathcal{W}_O$, interestingly, the optimal time required to extract the maximum work from the system also decreases monotonically with increasing value of $\gamma$, i.e., 
\begin{eqnarray}
t^{*}(\gamma) < t^{*}(\gamma^\prime) \quad \text{when} \quad \gamma > \gamma^\prime.    
\end{eqnarray}
The above investigation shows that although the maximal achievable work (operational ergotropy) decreases, it can be obtained faster compared to the noiseless case, see Fig. \ref{fig:alpha_space}(a) and (b). This also translates into the enhancement of the maximal output power,  $P^* (\gamma) = \mathcal{W}_O(\gamma)/t^*(\gamma)$ in the presence of environmental effects, see inset of Fig. \ref{fig:alpha_space} (b). Counterintuitively, when the work extraction is affected by ADC, the maximum attainable power becomes larger than the perfect unitary dynamics, however, this is not the situation for DPC and PDC. 


\subsubsection{Optimality of the geodesic drive for the considered Markovian noise models}
The last quantity of the triad $(\mathcal{W}_O,t^*, U_{\{\mu_k^*\}})$ that we investigate is the optimal unitary controlling the Schrodinger part of the dynamics $U_{\{\mu_k^*\}}$. 
In the presence of the considered Markovian noise models, our analysis reveals that the optimal drive for the Schrodinger part of the dynamics remains the one that was optimal in the noiseless case (See Appendix \ref{app:d} for details). 

The natural follow-up to this is to ask whether the continuity of optimality of the pacifying unitary in the noiseless case is a generic feature or not. In the subsequent section, we provide an explicit example of falsifying the above claim.

\subsection{Sub optimality of the geodesic drive}
\label{sec:geodesic}
\noindent
So far, 
we have demonstrated that the optimal unitary for the noiseless scenario coincides with the optimal unitaries in the noisy dynamics for some specific physical noise models. Nevertheless, this observation does not hold for arbitrary noisy dynamics. In this section, we provide a counterexample as a proof-of-principle demonstration of the sub-optimality of the geodesic drive. This turns out to be particularly important since it concretely establishes the importance of a more generalized notion of ergotropy, namely the operational ergotropy that we have introduced in Def. \ref{def:2}. Furthermore, the following example reinforces the necessity of considering the entire triad $(\mathcal{W}_O,t^*, U_{\{\mu_k^*\}})$ for a complete description of the work extraction.

Consider the evolution of a general state $\rho$ under the influence of noise, which can be described as  
\begin{eqnarray}
    \rho(t + \Delta t) = p(t) ~\mathcal{U}_{\Delta t}(\rho(t)) + (1-p(t)) \rho(t),
    \label{eq:noiselam}
\end{eqnarray}
where $\mathcal{U}_t(\cdot) = U(t) (\cdot) U^\dagger(t)$. Here, the probability $p(t) = 1 - \exp(-\zeta (1 - f_t)),$ with $f_t = \text{Tr} (\rho(t)\sigma)$ is the overlap of $\rho(t)$ with the attractor which is a fixed state $\sigma$ and $\zeta > 1$. This engineered noise model represents a completely positive trace preserving (CPTP) map. 
If the noiseless evolution trajectory passes near the attractor $\sigma$ the dynamics slows down considerably. Conversely, if the pacification is done through a trajectory that passes far away from $\sigma$, the noise effectively does not alter the noiseless dynamics. More precisely, a given evolution becomes insensitive to the noise when, throughout the trajectory (for all times $t$), $\zeta(1-f_t) >> 1$ is satisfied, which guarantees $p(t) \approx 1 ~\forall t$.
Now, when we consider a pacifying unitary generated by $H+V$, we have $U(\Delta t) = \mathbb{I} - i (H+V) \Delta t$. This translates to
\begin{eqnarray}
\frac{d \rho}{dt} =  -i [\mathcal{H}(t),\rho],
\label{eq:nl2}
\end{eqnarray}
where $\mathcal{H}(t) = p(t) (H + V)$, and $[A,B] = AB - BA$ is the commutator of $A$ and $B$. Note that the effect of the noise is to perturb the Hamiltonian in a time-dependent way. Furthermore, Eq. \eqref{eq:nl2} implies that the action of the noise retains the unitarity structure of the dynamics. However, this unitary is generated by the Hamiltonian  $\mathcal{H}(t)$ that essentially scales the noiseless Hamiltonian $H+V$ via the time-dependent scaling factor $p(t) \leq 1$. It effectively slows down the dynamics compared to the noiseless scenario, although the trajectory remains invariant.
As mentioned before, we recover the noiseless (unperturbed) evolution when $p(t) = 1 ~\forall t$.


To demonstrate the sub-optimality of the geodesic path, let us revisit the same example we have discussed in Sec. \ref{sec:2} where the initial state of the system is $\rho = |+\rangle \langle +|$, and $\omega_{\max} = 2$. The attractor of the noise model is taken to be $\sigma = \cos \psi \ket{0} + \sin \psi \ket{1}$ with $\psi = \pi/8$. 
From our previous analyses, we know that the set of pacifying unitaries is $U_\mu = \ketbra{0}{+} + e^{i\mu}\ketbra{0}{-}$, with the optimal drive given by $U_{\pi}$ which in turn traces out a quarter circle on the Bloch sphere from $\ket{+}$ to its passive state $\ket{0}$.
Now we are set to address the question of whether the optimal pacifying unitary in the noiseless case remains so in the noisy cases as well.
\begin{observation}
    The drive along the geodesic leads to sub-optimal work extraction.
\end{observation}
When the optimal unitary in the noiseless scenario, $U_\pi$ , a rotation about the axis $\bm{\hat{n}_{\pi}}=(0,-1,0)$, is used for work extraction in the above noisy model, the effective dynamics becomes asymptotically slow. The above phenomenon is due to the fact that $\sigma$ lies on the pacification trajectory,  the fixed point of the dynamics. The extractable work in turn saturates to  $\frac12 - \sin^2 \psi|_{\psi=\pi/8} = \frac{1}{2\sqrt{2}}$, as shown in Fig. \ref{fig:alpha_space} (c). The saturation time in this case is asymptotically long, i.e., $t_{U_{\pi}}^N \to \infty$, where $N$  indicates the presence of noise. Therefore, the geodesic drive turns out to be useless for work extraction. 
On the other hand, by considering our benchmark example involving a non-optimal pacifying unitary, denoted as $U_0$, which entails a rotation around $\bm{\hat{n}_0}= \frac{1}{\sqrt{2}}(1,0,1)$, we can achieve the desired passive state.
Hence, the maximum amount of extractable work coincides with $\mathcal{W}_O = 1/2$, and is achieved at $t_{U_{0}}^N \approx 3.026$. Thus, it is evident that $U_0$ performs better than the geodesic drive. Therefore, in the presence of noise, we obtain a reversed hierarchy of pacifying unitaries, i.e, $t^N_{U_\pi} > t^N_{U_0}=t^*$, whereas in the noiseless case, we had $t_{U_\pi} < t_{U_0}$. Recall that $t_{U_\pi}$ is the optimal time $t^*$ for the noiseless scenario. Finally, we show that $U_0$ turns out to be the optimal drive in this case. The details of the proof of the optimality are provided in the Appendix. \ref{app:d}. The complete work extraction triad for this noise model is $(\frac12, 3.026, U_0)$. 

\section{Conclusion}
\label{conclusion}
Work extraction is one of the primary goals of thermodynamics from its very inception. From steam engines to quantum batteries, effective extraction of work has turned out to be crucial in building technologies that run the world in both classical and quantum regimes. 
In the quantum domain,  one of the standard frameworks of work extraction is carried out via implementing unitary operations, referred to as ergotropy. Here we identify that there exists a family of unitaries, which we call pacifying unitaries, that lead to the same amount of extracted work. As a result, the entire set was deemed to be identical. However, when we limit the maximal energy that can be pumped into the system, all the pacifying unitaries typically become inequivalent in terms of the time required to extract the maximum work. Here we introduced the concept of optimality in which the pacifying unitary that extracts work equal to the ergotropy in the shortest amount of time is considered optimal. Moreover, 
we generalized the purview of ergotropy to situations where the unitaries used for ergotropic work extraction are tampered with noise. This led to a concept of operational ergotropic work that is extracted by generalized CPTP maps whose noiseless avatar corresponds to a pacifying unitary. Here we addressed two important questions: In the presence of noise, what can be the optimal choice of unitary for the Schrodinger part of the dynamics? and Is it the optimal pacifying unitary obtained in the noiseless case? 

For some paradigmatic Markovian noise models, we found that the best work extraction features are indeed obtained when the Schrodinger part of the dynamics is governed by the optimal pacifying unitary corresponding to the noiseless case. Interestingly, in these cases, although noise reduced the maximal work that can be extracted, it may be extracted faster compared to the noiseless case. However, we reported that the choice of optimal dynamics is not universal. Strikingly, we furnished an example of a noise model where the optimal pacifying unitary of the noiseless sector performs worse than other pacifying unitaries thereby establishing its sub-optimality in the generalized framework. 

Although the fundamental laws of nature are time translation invariant,  the irreversibility of thermodynamic processes results in a definite time direction. In this letter, we established that not only the arrow of time, but the amount of time turns out to be particularly important in thermodynamic work extraction processes. In both noiseless and noisy scenarios, where we generalized the concept of ergotropy, by focusing on time requirements, we prescribed how to filter out optimal work extraction strategies from a set of ``apparently" equivalent ones. We believe that our work would pave the way for future research exploring the time requirements in general work extraction strategies in various scenarios involving quantum correlations (operational daemonic ergotropy) \cite{Francica2017} or non-Markovianity. 


\section*{Acknowledgement}
We thank Aditi Sen (De) for useful discussions and for providing insights during the preparation of the manuscript.

\bibliographystyle{apsrev4-1}
\bibliography{bib.bib}

\appendix

\section{The normal to the plane of the geodesic}
\label{app:B}
To compute the unit normal $\hat{n}$, we first define two vectors connecting the origin of the Bloch sphere to  $\rho$ and $\rho_p$. \begin{eqnarray}
    \bm{V}_{\rho} &=& \{r \sin \theta_0 \cos \phi_0, r \sin \theta_0 \sin \phi_0, r \cos \theta_0\}, \nonumber \\
    \bm{V}_{\rho_p} &=& \{0,0,r\}.
\end{eqnarray}
These two vectors lie in the plane of the geodesic. Therefore, the normal to this plane is given by
\begin{eqnarray}
    {\bf{n} = \bf{V}_{\rho}} \times {\bm{V}_{\rho_p}} = r^2 \sin \theta_0 \{ \sin \phi_0, -\cos \phi_0,0\}.
\end{eqnarray}
The unit normal is 
\begin{eqnarray}
 {\bf{\hat{n}}} = \frac{\bm{n}}{||\bm{n}||} = \{ \sin \phi_0, -\cos \phi_0,0\}, 
 \label{eq:unitnormal}
\end{eqnarray}
where $||\bm{v}|| = \sqrt{\bm{v}.\bm{v}}$ is the norm of the vector $\bm{v}$.
This entirely resides in the $XY$ plane. The corresponding pacifying drive consistent with the energy bound is given by
\begin{eqnarray}
  U(t) = \exp (\frac{-i\omega_{\max}t}{2} {\bf{\hat{n}.}} \bm{\sigma}).
\end{eqnarray}
The extracted work as a function of time for the geodesic drive 
\begin{eqnarray}
   \nonumber W(t) &=& \text{Tr} (H \rho) - \text{Tr} (H U(t)\rho U^\dagger(t)) \\ &=& r\sin{\bigg(\frac{\omega_{\max }}{2}t\bigg)}\sin{\bigg(\theta_0-\frac{\omega_{\max} }{2}t\bigg)}.
   \label{erg_optimal}
\end{eqnarray}
We obtain work equal to the ergotropy for $t = t^* = \frac{\theta_0}{\omega_{\max}}$.

\section{Master equation for different channels}
\label{master_eqn}
For any state $\rho$, the dissipators of various noise Markovian noise models are as follows. 
\begin{eqnarray}
    \text{Amplitude damping:} &~&\gamma(2\sigma_-\rho\sigma_+ - \sigma_+\sigma_-\rho - \rho\sigma_+\sigma_-), \nonumber \\
    \text{Phase damping:} &~&\frac{\gamma}{2}(2\sigma_3\rho\sigma_3 - 2\rho), \nonumber \\
    \text{Depolarization channel:} &~&\frac{\gamma}{4}(2\sum_{i=1}^3\sigma_i\rho\sigma_i - 6\rho). \nonumber
\end{eqnarray}

Here, $\sigma_i$s are the usual Pauli matrices with $\sigma_{+}=\ketbra{1}{0},\sigma_{-}=\ketbra{0}{1}$ as raising and lowering operator respectively, and $\gamma$ is the decay rate.

\section{Generators for different Pacifying unitaries}
\label{app:C}
First we note that $K \hat{n}_\mu.\bm{\sigma} = H + V_\mu$ generates $U_\mu$. To obtain the generator, we expand and equate $e^{-i K \hat{n}_\mu.\bm{\sigma} t_\mu}$ to $U_\mu$. In particular, we have $ e^{-i K \hat{n}_\mu.\bm{\sigma}t_\mu} = U_\mu = \sum_{k = 0}^3 a_\mu^k \sigma_k$, where $a^k_\mu = \frac12 ~\text{Tr} (\sigma_k U_\mu)$. We now can easily solve for $t_\mu$ and $\hat{n}_\mu$  from the following expansion by equating each coefficient of $\sigma_i$s 
\begin{eqnarray}
\cos (K t^*) ~\sigma_0 - i \sin (K t^*) ~\hat{n}_\mu.\bm{\sigma} =  \frac12 ~\text{Tr} (\sigma_k U_\mu).
\end{eqnarray}
Moreover, following the energy bound in Eq. \eqref{eq:energybound1} of the main text, we can substitute $K = \frac{\omega_{\max}}{2}$.


\section{Optimal drives in presence of noise}
\label{app:d}
\begin{figure}[ht]
\vspace*{0.2in}
\includegraphics[width=0.9\linewidth]{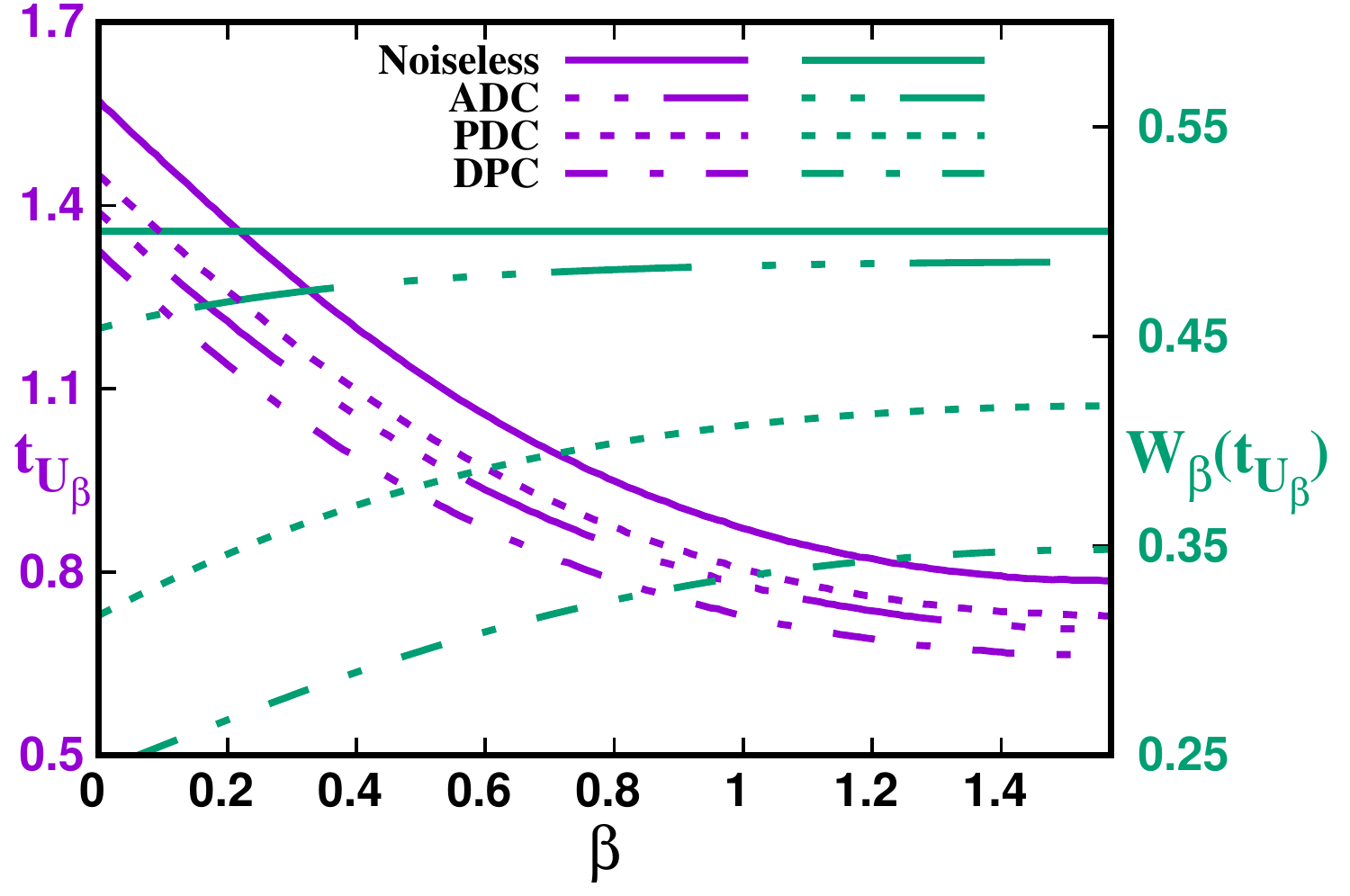}
\caption{(Color online.) Minimal pacification time $t_{U_{\beta}}$ (left ordinate, violet color) and work extraction $W_{\beta}(t_{U_{\beta}})$ (right ordinate, green color) are plotted as a function of $\beta$ by taking initial state $\ket{+}$ and $\omega_{\max}=2$, $\gamma=0.25$. Solid lines correspond to noiseless scenario whereas the dashed lines signify prototypical noise models , i.e., ADC, PDC and DPC. Noiseless optimal unitary $U_{\beta=\pi/2}\equiv U_{\mu=\pi}$ remains the optimal one for Schrodinger part of the dynamics in the mentioned noisy evolution. All the axes are dimensionless.}
\label{fig:optimzation} 
\end{figure}
A set of pacifying unitaries was identified for the considered initial state $\rho = \ket{+}\bra{+}$ drives $\rho$ to its corresponding passive state $\ket{0}\bra{0}$. Geometrically, these unitaries pacify through various distinct trajectories on the Bloch sphere. These trajectories are circular paths obtained from the intersection of the plane passing through the two points on the Bloch sphere corresponding to states $\ket{+}$ and $\ket{0}$, and the Bloch sphere. Each such pacification trajectory is completely characterized by the particular plane induced by a given pacifying unitary.
The planes connecting $\ket{+}$ and $\ket{0}$ are characterized  by a one parameter family of planes 
\begin{eqnarray}
    \bf{\hat{n}}_\beta =\Big(\frac{\cos\beta}{\sqrt{2}},-\sin\beta,\frac{\cos\beta}{\sqrt{2}}\Big),
\end{eqnarray}
where $\beta\in[0,\pi/2]$.
Note that for a complete characterization, we need to consider $\bf{\hat{n}}_{-\beta}$ as well. However, the optimal drives, at least for the considered noise models come from the drives generated $\bf{\hat{n}}_\beta$. The time dependent drive along $\bf{\hat{n}}_\beta$ is 
\begin{eqnarray}
    U_{\beta}(t)=\exp\left(-it\bf{\hat{n}}_\beta .\bm{\sigma} \right).
\end{eqnarray}
Now the quantities of interest are $t_{U_\beta}$, which is minimal time required for work extraction for the drive along $\bf{\hat{n}}_\beta$ and the corresponding extracted work $W_{\beta}(t_{{U}_{\beta}})$.
\subsection{Optimality of geodesic drive for paradigmatic Markovian noise models}
For the Markovian noise models considered in this work, by optimizing over $\beta$ as shown in Fig. \ref{fig:optimzation} we find that the minimal time $t_{U_\beta}$ required for work extraction is achieved for $\beta = \pi/2$. At this time the extracted work is maximized simultaneously. 
Physically, this corresponds to a rotation about the $-y$ axis. It in turn corresponds to $U_{\mu = \pi}$ drive 
considered in Sec.~\ref{markovian}. Therefore, we get $t^*=t_{U_{\beta=\pi/2}} = t_{U_{\mu=\pi}}^N$. This analysis holds true for other choices of transition rates $\gamma$ as well.

\subsection{Demonstrating the suboptimality of geodesic drive}
For the case of the noise model considered in Sec.~\ref{sec:geodesic},  the search for optimality is demonstrated in Fig.~\ref{fig:handmade_optimzation}. 
\begin{figure}[ht]
\includegraphics[width=0.9\linewidth]{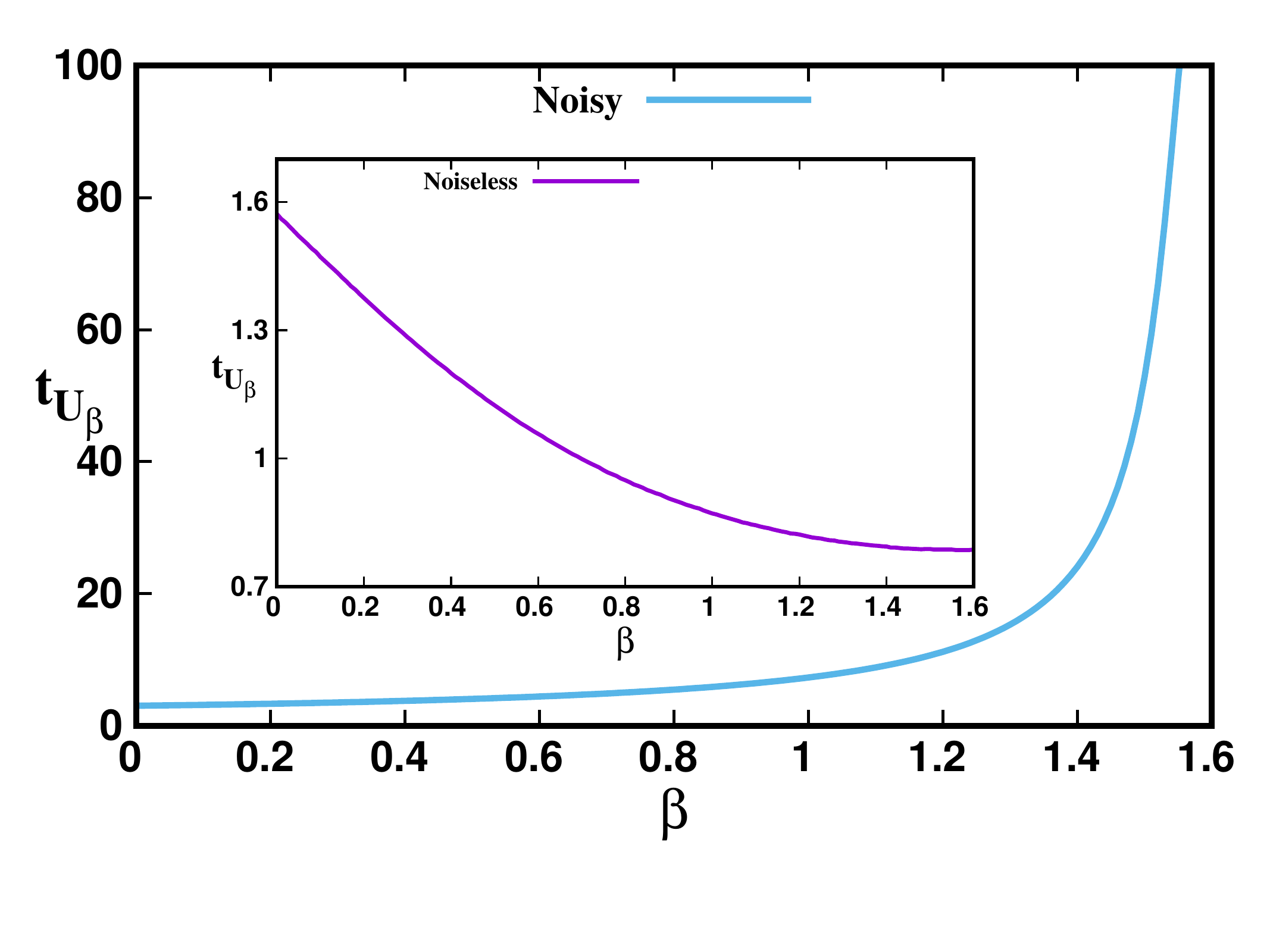}
\caption{(Color online.) $t_{U_\beta}$ (vertical axis) is plotted against $\beta$ (horizontal axis) where the initial state is $\ket +$ with $\omega_{\max}=2$ and $\zeta = 5$ for the noise model considered in Sec. \ref{sec:geodesic}. Optimal unitary in this scenario becomes $U_{\beta=0}=U_{\mu=0}$. In the inset we have compared the same for the unperturbed evolution. All the axes are dimensionless.}
\label{fig:handmade_optimzation} 
\end{figure}
Here, unlike in the previous case, $\beta = 0$ corresponds to the best work extraction scheme. In particular, we find $t_{U_{\beta=0}} \approx 3.026$, is the optimal path for noise model in Eq.~\eqref{eq:noiselam}. Note that $\beta = 0$ implies rotation along the axis $\bm{\hat{n}}= \frac{1}{\sqrt{2}}(1,0,1)$ as claimed in Sec. \ref{sec:geodesic}.  Therefore, we have $t^*=t_{U_{\beta=0}} = t_{U_{\mu=0}}^N$.

\end{document}